\documentclass{article}
\usepackage[letterpaper]{geometry}

\usepackage[dvipsnames]{xcolor}
\usepackage{amsmath}
\usepackage{amssymb}
\usepackage{bbm}
\usepackage{cases}
\usepackage{hyperref, cleveref, graphicx,subfigure}

\usepackage{amsthm}
\usepackage[dvipsnames]{xcolor}

\newtheorem{theorem*}{Theorem}
\newtheorem{lemma*}[theorem*]{Lemma}
\newtheorem{definition}[theorem*]{Definition}
\newtheorem{corollary*}[theorem*]{Corollary}

\usepackage{tikz}
\usetikzlibrary{arrows.meta}
\tikzset{
  arr/.style={{}-Latex,shorten <=-2pt}
}
\usetikzlibrary{automata,calc,positioning}

\newcommand{\node}{\mbox{\rm node}}

\newcommand{\inp}{\mbox{\rm inp}}

\newcommand{\KA}{{\rm \textsf{KNOW-ALL}}}
\newcommand{\tr}{\mathbf{T}}

\title{\Large A Simple Lower Bound for Set Agreement in Dynamic Networks}
\author{Pierre Fraigniaud\thanks{Institut de Recherche en Informatique Fondamentale (IRIF), CNRS and Université Paris Cité, France. Additional support from  ANR projects DUCAT (ANR-20-CE48-0006) and QuDATA (ANR-18-CE47-0010).}
\and Minh Hang Nguyen\thanks{Institut de Recherche en Informatique Fondamentale (IRIF), CNRS and Université Paris Cité, France. Additional support from  ANR projects DUCAT (ANR-20-CE48-0006), TEMPORAL (ANR-22-CE48-0001), and the European Union’s Horizon 2020 program H2020‑MSCA ‑COFUND‑2019 Grant agreement n° 945332.}
\and Ami Paz\thanks{Interdisciplinaire des Sciences du Numérique (LISN), CNRS and Université Paris-Saclay, France}
}

\date{}

\begin{document}

\maketitle

\begin{abstract} \small\baselineskip=9pt
Given a positive integer $k$, $k$-set agreement is the distributed task in which each process $i\in [n]$ in a group of $n$ processing nodes starts with an input value $x_i$ in the set $\{0,\dots,k\}$, and must output a value~$y_i$ such that (1)~for every $i\in[n]$, $y_i$~is the input value of some process, and (2)~$|\{y_i : i\in [n]\}|\leq k$. That is, at most $k$ different values  in total  must be outputted by the processes. The case $k=1$ correspond to (binary) consensus, arguably the most studied problem in distributed computing. While lower bounds or impossibility results for consensus have been obtained for most of the standard distributed computing models, from synchronous message passing systems to asynchronous shared-memory systems, the design of lower bounds or impossibility results in the case of $k$-set agreement for $k>1$ is notoriously known to be much more difficult to achieve, and actually remains open for many models. 
The main techniques for designing lower bounds or impossibility results for $k$-set agreement with $k>1$ use tools from algebraic topology. 

The algebraic topology tools are difficult to manipulate, and require lot of care for avoiding mistakes in their applications. This difficulty increases even more whenever the communications are mediated by a network of arbitrary structure, such as in the standard LOCAL model for distributed network computing. Recently, the $\KA$ model has been specifically designed as a first attempt to understand the LOCAL model through the
lens of algebraic topology, and Castañeda et al.~(2021) have designed lower bounds for consensus and $k$-set agreement in the $\KA$ model, with applications to dynamic networks. The arguments for these lower bounds are extremely sophisticated, and uses high-dimensional connectivity. 
Even if the sequence of evidences can be convincingly followed, a typical reader may have difficulties to get the overall intuition or to reconstruct the proof. 

In this work, we re-prove the same lower bound for $k$-set agreement in the $\KA$ model. This new proof stands out in its simplicity, which makes it accessible to a broader audience, and increases confidence in the result.
\end{abstract}

\section{Introduction}
\label{sec:Intro}

We study $k$-set agreement problem in the synchronous message-passing model for networks. In this model, a set of $n\geq 2$ processes perform in lockstep, as a sequence of synchronous rounds. At each round, each process performs some individual computation, and communicates once with its neighbors in the network. Specifically, the network is modeled as an $n$-node graph $G=(V,E)$ in which every node is a process. At each round, every node can send messages to its neighbors in~$G$, receive messages from these neighbors, and perform some computation. The complexity of a distributed algorithm is measured as its number of rounds until all nodes halt. This basic setting can be trivially extended to the case in which the communication graph changes in each round, for capturing computation in dynamic networks. The $n$ processes are labeled by integers from~1 to~$n$, and we let $[n]=\{1,\dots,n\}$. The paper is focusing on the round-complexity of $k$-set agreement in dynamic networks.  

\begin{definition} For every integer $k\geq 1$, the $k$-set agreement task is the following. 
    Every process $i\in [n]$ is assigned an input value  $x_i\in \{0,\ldots,k\}$, and must output a value  $y_i\in \{0,\ldots,k\}$ such that the following two conditions hold: 
    \begin{itemize}
        \item Validity: For every process $i\in[n]$, its output value $y_i$ is the input value $x_j$ of at least one process $j\in [n]$;
        \item Agreement: The number of different output values is at most~$k$, that is $|\{y_i:i\in [n]\}|\leq k$.
    \end{itemize}
\end{definition}

\subparagraph{Distributed Computing through the Lens of Algebraic Topology.} 

The $k$-set agreement task has been studied in essentially all standard models of distributed computing, from synchronous message-passing systems to asynchronous shared-memory systems. While the design of efficient $k$-set agreement algorithms is reasonably easy in most models, establishing their optimality is often quite challenging. The main, and actually unique tools at our disposal for establishing lower bounds (or impossibility result) for $k$-set agreement are from algebraic topology~\cite{HerlihyKR2013}. In a nutshell, the communication occurring in a distributed system (under the full-information paradigm) can be modeled  as a topological deformation of a simplicial complex formed by all possible input configurations of the system, and the computation (i.e., the decision of each process regarding its output value) can be modeled as a simplicial map from the deformed input complex to a simplicial complex formed by all possible output configurations of the system. In addition, the simplicial map must respect the input-output specification of the task. 

The topological deformation of the input complex strongly depends on the communication model, and designing lower bounds for a task, e.g.,  $k$-set agreement, requires to understand what is the smallest number of rounds such that there exists a simplicial map from the deformed input complex to the output complex, respecting the input-output specification of the task. Roughly, solving (binary) consensus, i.e., solving $k$-set agreement for $k=1$, requires that communications eventually disconnect the input complex. The intuition is that the output complex is merely composed of two disjoint simplices, one corresponding to all-0 outputs, and another corresponding to all-1 outputs. The situation is more intricated for $k$-set agreement with $k>1$. To see why, let us recall the so-called \KA\/ model. 

\subparagraph{Application to Distributed Network Computing.}

One of the most popular models for understanding physical locality in distributed computing is the \textsf{LOCAL} model~\cite{Peleg2000}.
The \KA\/ model has been designed by Castañeda, Fraigniaud, Paz, Rajsbaum, Roy, and Travers~\cite{castaneda2021topological}
as a first attempt to study the \textsf{LOCAL} model through the lens of algebraic topology. In the \KA\/ model, the $n$ (fault-free) processes communicate synchronously via a sequence $\mathcal{S}=(G_t)_{t\geq 1}$ of directed graphs on the same set $[n]$ of nodes, where communication occurring at round~$t\geq 1$ is performed along the arcs of~$G_t$. Importantly, each process~$i\in [n]$ occupies node~$i$ of each of the graphs $G_t$, $t\geq 1$, and processes are aware of the sequence $\mathcal{S}$ of communication graphs (hence the name \KA\/ for the model). The only uncertainty of a node is about the inputs given to the other nodes. The simple question asked in~\cite{castaneda2021topological} is: 

\medbreak
\centerline{\it What is the round-complexity of solving $k$-set agreement in $\mathcal{S}=(G_t)_{t\geq 1}$?}
\medbreak

\noindent 
For answering this question, Castañeda et al. defined a family of directed graphs capturing the way the information flow in the dynamic network $\mathcal{S}$. Given $k\geq 1$ and $\mathcal{S}=(G_t)_{t\geq 1}$, let us define, for every $r\geq 1$, the digraph  $H_r=H_r(k,\mathcal{S})$ with vertex set $[n]$ such that there is an arc $(u,v)$ from $u$ to $v$ in $H_r$ if there is a sequence of nodes $w_1,\dots,w_r,w_{r+1}$ such that $w_1=u$, $w_{r+1}=v$, and, for every $t\in\{1,\dots,r\}$, $w_t=w_{t+1}$ or $(w_t,w_{t+1})\in E(G_t)$. Hence, in particular, if all the graphs $G_t$ in $\mathcal{S}$ are the same graph~$G$, then $H_r$ is merely the $r$-th transitive closure of~$G$, that is, there is an arc from $u$ to $v$ in $H_r$ whenever the distance from $u$ to $v$ is $G$ is at most~$r$. 

For any graph or digraph~$H$, let $\gamma(H)$ denote the \emph{dominating number} of $H$, that is, the minimum size of a dominating set in~$H$. Castañeda et al. have shown that $k$-set agreement requires $r$ rounds, where $r$ is the smallest integer such that there exists a $k$-node dominating set in $H_r$. Note that there is a trivial algorithm solving $k$-set agreement in $\mathcal{S}$, in $r$ rounds. Let $i_1,\dots,i_q$ be $q\leq k$ nodes forming a dominating set in $H_r$. At the first round, every node $i\in [n]$ sends the pair $(i,x_i)$ to all its out-neighbors in~$G_1$, where $x_i$ denotes the input of node~$i$. At each of the subsequent rounds $t>1$, every node forwards to its out-neighbors in $G_t$ all the pairs received from its in-neighbors in $G_{t-1}$ during the previous round. After $r$ rounds of this flooding process, every node $i\in [n]$ has received at least one pair $(j,x_j)$ with $j\in\{i_1,\dots,i_q\}$, as $\{i_1,\dots,i_q\}$ dominates~$H_r$. Node $i$ merely outputs~$x_j$ (where $j\in\{i_1,\dots,i_q\}$ is picked arbitrarily if node~$i$ received several pairs $(j,x_j)$ with $j\in\{i_1,\dots,i_q\}$). 

The correctness of this simple algorithm follows from the fact that $\gamma(H_r)\leq k$. Interestingly, it is far from being trivial to establish that the round-complexity of the algorithm is actually optimal. Theorem~4.1 and Corollary~5.6 in \cite{castaneda2021topological} can be restated as follows. 

\begin{theorem*}[Castañeda, Fraigniaud, Paz, Rajsbaum, Roy, and Travers~\cite{castaneda2021topological}]\label{theo:main}
Let $k\geq 1$ be an integer, and let $\mathcal{S}=(G_t)_{t\geq 1}$ be an instance of the \KA\/ model. 
Every algorithm solving $k$-set agreement in $\mathcal{S}$ requires at least $r$ rounds where $r$ be the smallest integer such that $\gamma(H_r)\leq k$.
\end{theorem*} 

To get an intuition of this result, let us consider the basic case of consensus, i.e., $k=1$. 

\subparagraph{Lower bound for consensus.} 

For $k=1$, the proof of the lower bound is by a standard indistinguishability argument. Specifically, let us assume, for the purpose of contradiction, that there is an algorithm ALG solving consensus in $\mathcal{S}=(G_i)_{i\geq 1}$ in $r$ rounds, where $r$ is an integer such that $\gamma(H_r) > 1$. Let us consider the input configuration~$I_0$ in which all nodes have input~0. For every $i=1,\dots,n$, we gradually change the input configuration as follows (see Figure~\ref{fig:consensus-one-dim}). For $i=1,\dots,n$, given configuration $I_{i-1}$, the configuration $I_i$ is obtained from $I_{i-1}$ by switching the input of node~$i$ from~0 to~1. Note that $I_n$ is the input configuration in which all nodes have input~1. 
Since $\gamma(H_r) > 1$, for every $i = 1,\ldots,n$, there exists a node~$w_i$ that has not received the input of node~$i$ in ALG after $r$ rounds. In particular, for every $i\in\{1,\dots,n\}$, node $w_i$ does not distinguish $I_{i-1}$ from $I_i$.  Therefore, for every $i\in [n]$, ALG must output the same at $w_i$ in both input configurations $I_{i-1}$ and $I_i$. Since, for every $i\in\{0,\dots,n\}$, all nodes must output the same value for input configuration~$I_i$, we get that the consensus value returned by ALG for $I_0$ is the same as the consensus value returned by ALG for~$I_n$, which contradicts the validity condition.

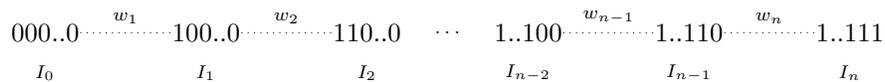
\begin{figure}[tb]
    \centering
    \begin{tikzpicture}[scale = 1.07]

\tikzstyle{whitenode}=[circle,minimum size=0pt,inner sep=0pt,font=\scriptsize]
\tikzstyle{thicknode}=[circle,minimum size=0pt,inner sep=0pt]

\draw (0,0) node[thicknode] (a0)   [] {000..0};
\draw (2,0) node[thicknode] (a1)   [] {100..0};
\draw (4,0) node[thicknode] (a2)   [] {110..0};
\draw (5,0) node[whitenode] ()   [] {$\ldots$};
\draw (6,0) node[thicknode] (a3)   [] {1..100};
\draw (8,0) node[thicknode] (a4)   [] {1..110};
\draw (10,0) node[thicknode] (a5)   [] {1..111};

\draw (0,-0.5) node[whitenode] ()   [] {$I_0$};
\draw (2,-0.5) node[whitenode] ()   [] {$I_1$};
\draw (4,-0.5) node[whitenode] ()   [] {$I_2$};
\draw (6,-0.5) node[whitenode] ()   [] {$I_{n-2}$};
\draw (8,-0.5) node[whitenode] ()   [] {$I_{n-1}$};
\draw (10,-0.5) node[whitenode] ()   [] {$I_n$};

\path[dotted,draw] (a0) edge node[above,font=\scriptsize] {$w_1$} (a1);
\path[dotted,draw] (a1) edge node[above,,font=\scriptsize] {$w_2$} (a2);
\path[dotted,draw] (a3) edge node[above,,font=\scriptsize] {$w_{n-1}$} (a4);
\path[dotted,draw] (a4) edge node[above,,font=\scriptsize] {$w_n$} (a5);

\end{tikzpicture}
    \caption{Input configurations $I_0,\ldots,I_n$. In $I_0$, all nodes have input $0$. For every $i\in [n]$, configuration $I_i$ is obtained from $I_{i-1}$ by changing input of node $i$ from $0$ to $1$. For two consecutive input configurations $I_{i-1}$ and $I_i$, there is exactly one node $i\in [n]$ with different inputs, and a node $w_i$ cannot distinguish $I_{i-1}$ from $I_i$.} 
    \label{fig:consensus-one-dim}
\end{figure}

Generalising this standard indistinguishability argument to the case $k \geq 2$ is non trivial. Instead, Castañeda at al. \cite{castaneda2021topological} established Theorem~\ref{theo:main} by analysing the topological deformation of the input complex, and the homotopy of this deformed complex, as it is known that the level of connectivity of this simplicial complex is inherently related to the ability of solving $k$-set agreement (cf. Theorem 10.3.1 in \cite{HerlihyKR2013}). The analysis in \cite{castaneda2021topological} is quite sophisticated, but also hard to follow. That is, even if the sequence of evidences in the proof can be convincingly followed, a typical reader may have difficulties to get the overall picture of the proof. 

\subsection{Our Results}

The contribution of this paper is a new proof Theorem~\ref{theo:main}. We believe that our proof is much more self-contained than the one in \cite{castaneda2021topological}. It should be accessible to a wider audience, non necessarily aware of subtle notions of algebraic topology.

In fact, we believe that the proof for $k=2$ (the simplest interesting case) is simple enough to be taught in class as an example of using topological arguments in distributed computing. As opposed to the complex proof in \cite{castaneda2021topological} the use of algebraic topology will be limited to applying Sperner's Lemma. 

To get an intuition, recall that Sperner's Lemma states that every Sperner coloring of a triangulation of a $k$-dimensional simplex contains a cell whose vertices all have different colors. For instance, for $k=2$, Sperner's Lemma refers to a triangulation of a triangle $(A,B,C)$ where (1)~$A,B$, and $C$ are respectively colored $0,1$, and $2$, (2)~all vertices of $AB,AC$, and $BC$ are respectively colored arbitrarily with colors in $\{0,1\}$, $\{0,2\}$, and $\{1,2\}$, and (3)~all internal vertices are colored arbitrarily with colors in $\{0,1,2\}$. Sperner's Lemma states that there is a cell whose vertices all have different colors $0,1$, and~$2$. Here is how to use Sperner's Lemma. Let ALG be an algorithm that is supposed to solve $k$-set agreement, if one can associate a process to each vertex of the triangulation, and if one can then associate the output value of ALG at this process to this vertex, in such a way that every cell matches with one input configuration, then Sperner's Lemma will establish that there is indeed an input configuration for which ALG returns $k+1$ different output values, which establishes that ALG violates agreement. In other words, our proof merely uses the well known Sperner's Lemma, whereas the proof of~\cite{castaneda2021topological} uses complex arguments related to shellability and Nerve Lemma. In our proof, we consider a specific subset of input configurations, rather than considering all possible input configurations as represented by the input complex in~\cite{castaneda2021topological}.

\subsection{Related work}

As for many domains of computer science, the design of lower bounds on the complexity of the tasks at hand is often challenging in distributed computing. In this section, we shall restrict ourselves to surveying some of the main contributions related to establishing lower bounds on the complexity of $k$-set agreement in various forms of distributed network models, and we shall let aside all the efforts dedicated to other models such as the standard asynchronous shared-memory model. We refer to~\cite{HerlihyKR2013} for these latter settings, and ~\cite{attiya2013non,10.1016/j.jpdc.2015.09.002} for simplified topological proofs.

In synchronous fault-prone networks, the study of $k$-set agreement remained for a long time mostly confined to the special case of the message-passing model in the \emph{complete networks}. That is, $n$ nodes subject to \emph{crash} or \emph{malicious} (a.k.a.~Byzantine) failures are connected as a complete graph~$K_n$ in which every pair of nodes has a private reliable link allowing them to exchange messages.  In this setting, a significant amount  of effort has been dedicated to narrowing down the complexity of solving agreement tasks such as consensus and, more generally, $k$-set agreement for $k\geq 1$. 
For $t\geq 0$, the standard \emph{synchronous $t$-resilient message-passing} model assumes $n\geq 2$ nodes labeled from 1 to $n$, and connected as a clique, i.e., as a complete graph~$K_n$. Computation proceeds as a sequence of synchronous rounds. Up to $t$ nodes may crash during the execution of an algorithm. When a node $v$ crashes at some round~$r\geq 1$, it stops functioning after round $r$ and never recovers. 
Moreover, some  (possibly all) of the messages sent by $v$ at round~$r$ may be lost. 
This model has been extensively studied in the literature \cite{attiya2004distributed,HerlihyKR2013,Lynch96,Raynal2010}. In particular, it is known that consensus can be solved in $t+1$ rounds in the $t$-resilient model~\cite{dolev1983authenticated}, and this is optimal for every $t<n-1$ as far as the worst-case complexity is concerned~\cite{aguilera1999simple,dolev1983authenticated}. Main techniques used to show this consensus lower bound is by indistinguishability arguments, e.g., using a notion called \emph{valency}~\cite{aguilera1999simple}, or equivalent relation on set of histories~\cite{dolev1983authenticated}. 
Similarly, $k$-set agreement, in which the cardinality of the set of output values decided by the (correct) nodes must not exceed~$k$, is known to be solvable in $\lfloor t/k\rfloor+1$ rounds~\cite{chaudhuri1991towards}, and this worst-case complexity is also optimal~\cite{chaudhuri1993tight}. The result was first proved using Sperner Lemma, and Kuhn's triangulation~\cite{chaudhuri1993tight}. 

It is only very recently that the synchronous $t$-resilient message-passing model has been extended to the setting in which the complete communication graph $K_n$ is replaced by an arbitrary communication graph~$G$ (see~\cite{CastanedaFPRRT23,ChlebusKOO23}). We point out that our lower bound for $k$-set agreement in arbitrary graphs is inspired from the lower bound for  $k$-set agreement in complete  graphs in~\cite{chaudhuri1993tight}.
It is also worth mentioning a  follow-up work~\cite{FraigniaudP20} aiming at minimizing the size of the simplicial complexes involved in the analysis of distributed computing in networks, which extended the framework for handling graph problems such as finding a proper coloring.

The \emph{oblivious message adversary} model allows an adversary to choose each communication graph $G_t$ from a set $\mathcal{G}$, independently of its choices for the other graphs $G_s$ for $s<t$.
The nodes know the set $\mathcal{G}$ a priori, but not the actual graph picked by the adversary at each round. We refer to \cite{coulouma2013characterization,nowak2019topological,winkler2024time} for recent advances in this domain, including solving consensus.
Consensus lower bounds are established by $\beta$-relation in~\cite{coulouma2013characterization}, refinements of a indistinguishability graph~\cite{winkler2024time}, path-connectivity of incompatible borders of a protocol complex~\cite{galeana2023topological}. 
Another powerful technique developed for studying lower bound is from combinatorial topology. In particular, some useful tools are Nerve lemma and shellability~\cite{herlihy2013topology, shimi2020k}. These tools help to show connectivity of a deformed input complex after a certain number of rounds. In~\cite{shimi2020k}, the ``closed-above'' model is first introduced, that is the oblivious message adversaries model with a special set of allowed graphs, and a lower bound for $k$-set agreement is established. 

\section{Intuition for the Proof}

Theorem~\ref{theo:main} holds whenever the set of input values for $k$-set agreement is not $\{0,\dots,k\}$ but is of cardinality at least~$n$, say the set $[n]$. Indeed, let us assume that $k$-set agreement is solvable in $\mathcal{S}=(G_t)_{t\geq 1}$ by an algorithm ALG running in $r-1$ rounds, where $r$  is the smallest integer such that $\gamma(H_r)\leq k$. 
Let us then consider the instance of $k$-set agreement in which, for every $i\in [n]$, node~$i$ starts with input~$x_i=i$, and let $\{i_1,\dots,i_q\}$ be the set of output values returned by ALG for this instance, where $1\leq q\leq k$. By the definition of~$r$, there exists a node $a \in [n]$ that is not dominated by any of the $q$ nodes $i_1,\dots,i_q$. Let $b=y_a \in \{i_i,\dots,i_q\}$ denote node~$a$'s output, and let us consider another instance of $k$-set agreement, where the input $x_b$ of node~$b$ is replaced by any value distinct from~$b$. Since  node $a$ is not dominated by node~$b$, node~$a$ cannot distinguish the two executions of ALG in the two instances. Therefore, node~$a$ still outputs $b$ in the second instance, which violates the validity condition. 
The difficulty of establishing the lower bound increases when the size of the input set decreases. 
We show that the bound holds even if the input value $x_i$ of each node~$i\in [n]$ belongs to $\{0,\dots,k\}$, that is, to the smallest set for which $k$-set agreement is non trivial: it is solely required that the processes ``eliminate'' one value out of the $k+1$ possible input values.  
For this purpose, we use the so-called Kuhn's triangulation~\cite{chaudhuri1993tight,todd2013computation}. 

\subparagraph{Kuhn's triangulation.}

The Kuhn's triangulation in dimension~$k$ is a partition of a region of $\mathbb{R}^k$ into $k$-dimensional simplices (i.e., triangles for $k=2$, tetrahedrons for $k=3$, etc.). One example of a Kuhn's triangulation is illustrated on Fig.~\ref{fig:triangulation}, in which a triangle in $\mathbb{R}^2$ is subdivided into 25 triangles. 

To show how to use the concept of Kuhn's triangulation to establish our lower bound, let us consider the illustrative case in which $\mathcal{S}$ is the fixed 5-node graph $\vec{C}_5$ depicted on Fig.~\ref{fig:triangulation}, that is, the 5-node directed cycle, and let us show that 2-set agreement is not solvable in just one single round in~$\vec{C}_5$. (We have $H_1=\vec{C}_5$, and thus $H_1$ does not have a dominating set of size~2.)

Fig.~\ref{fig:triangulation}(a) shows the Kuhn's triangulation in which a input configuration $\inp(\mathbf{x})$ is associated to each vertex $\mathbf{x}$ of the triangulation. In the figure, a sequence $x_1x_2x_3x_4x_5$ denotes the input configuration in which, for every $i\in [5]$, node~$i$ receives input~$x_i$. Note that the path at the bottom of the triangulation of Fig.~\ref{fig:triangulation}(a) is a path of configurations $I_0,\dots,I_n$ similar to the one we used before to show the lower bound for consensus displayed in Fig.~\ref{fig:consensus-one-dim}. This assignment of input configurations to the vertices of the Kuhn's triangulation has a very important property. For each vertex $\mathbf{x}$ of the triangulation, there is a node $i\in [5]$, denoted $\node(\mathbf{x})$ satisfying that,  in just one round, $\node(\mathbf{x})$ cannot distinguish the configuration $\inp(\mathbf{x})$ from the configurations associated to its neighbor on the left, its neighbor on the bottom-left, and its neighbor below. Fig.~\ref{fig:triangulation}(b) displays the Kuhn's triangulation with the node assigned to each vertex of the triangulation. 

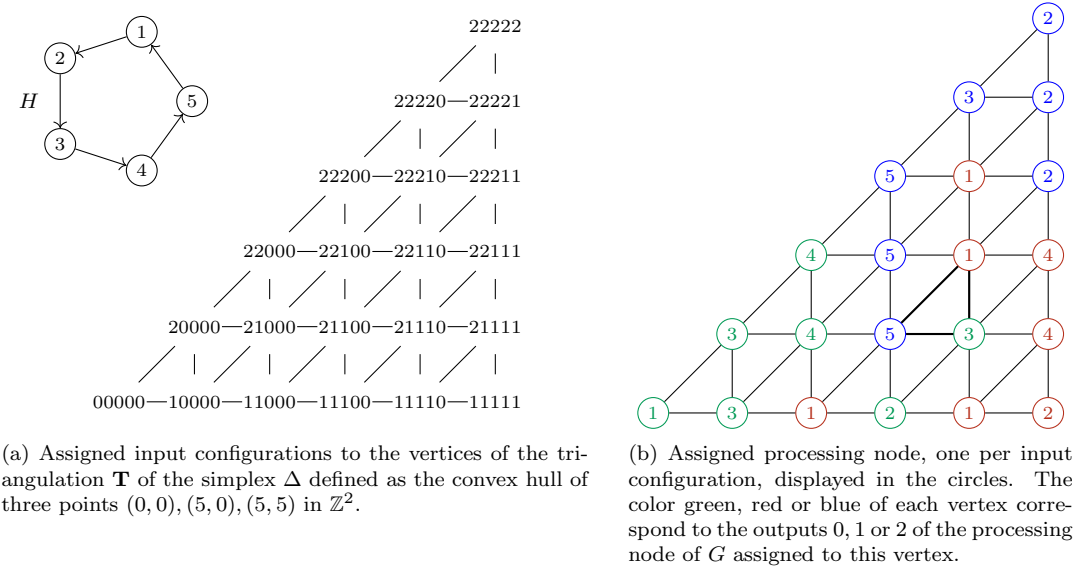
\begin{figure}
    \centering
\subfigure[Assigned input configurations to the vertices of the triangulation $\tr$ of the simplex $\Delta$ defined as the convex hull of three points $(0,0),(5,0),(5,5)$ in $\mathbb{Z}^2$.]{
\centering
\begin{tikzpicture}
\begin{scope}
    \foreach \phi in {1,...,5}{
    \node[draw,circle,fill=white,,inner sep=2pt, font=\scriptsize] (\phi) at ($(360/5 * \phi:0.97cm)+(0,4)$) {$\phi$};}
\draw [->] (1)--(2);
\draw [->] (2)--(3);
\draw [->] (3)--(4);
\draw [->] (4)--(5);
\draw [->] (5)--(1);
\draw (-1.2,4) node[font=\footnotesize] (H)   [] {$H$};
\draw (6,0) node[] (fake)   [] {};
\end{scope}

\tikzstyle{whitenode}=[circle,minimum size=0pt,inner sep=0pt,font=\scriptsize]

\draw (0,0) node[whitenode] (a0)   [] {00000};
\draw (1,0) node[whitenode] (a1)   [] {10000};
\draw (2,0) node[whitenode] (a2)   [] {11000};
\draw (3,0) node[whitenode] (a3)   [] {11100};
\draw (4,0) node[whitenode] (a4)   [] {11110};
\draw (5,0) node[whitenode] (a5)   [] {11111};

\draw (1,1) node[whitenode] (b1)   [] {20000};
\draw (2,1) node[whitenode] (b2)   [] {21000};
\draw (3,1) node[whitenode] (b3)   [] {21100};
\draw (4,1) node[whitenode] (b4)   [] {21110};
\draw (5,1) node[whitenode] (b5)   [] {21111};

\draw (2,2) node[whitenode] (c2)   [] {22000};
\draw (3,2) node[whitenode] (c3)   [] {22100};
\draw (4,2) node[whitenode] (c4)   [] {22110};
\draw (5,2) node[whitenode] (c5)   [] {22111};

\draw (3,3) node[whitenode] (d3)   [] {22200};
\draw (4,3) node[whitenode] (d4)   [] {22210};
\draw (5,3) node[whitenode] (d5)   [] {22211};

\draw (4,4) node[whitenode] (e4)   [] {22220};
\draw (5,4) node[whitenode] (e5)   [] {22221};

\draw (5,5) node[whitenode] (f5)   [] {22222};

\draw (a0)--(a1)--(a2)--(a3)--(a4)--(a5);
\draw (b1)--(b2)--(b3)--(b4)--(b5);
\draw (c2)--(c3)--(c4)--(c5);
\draw (d3)--(d4)--(d5);
\draw (e4)--(e5);

\draw (a5)--(b5)--(c5)--(d5)--(e5)--(f5);
\draw (a4)--(b4)--(c4)--(d4)--(e4);
\draw (a3)--(b3)--(c3)--(d3);
\draw (a2)--(b2)--(c2);
\draw (a1)--(b1);

\draw (a0)--(b1)--(c2)--(d3)--(e4)--(f5);
\draw (a1)--(b2)--(c3)--(d4)--(e5);
\draw (a2)--(b3)--(c4)--(d5);
\draw (a3)--(b4)--(c5);
\draw (a4)--(b5);

\end{tikzpicture}
}
\hspace*{.3cm} 
\subfigure[Assigned processing node, one per input configuration, displayed in the circles. The color green, red or blue of each vertex correspond to the outputs $0,1$ or $2$ of the processing node of $G$ assigned to this vertex.]{
\centering
\begin{tikzpicture}[scale = 1.05]
    \tikzstyle{whitenode}=[draw,circle,ForestGreen,minimum size=0pt,inner sep=2pt,font=\scriptsize]
    \tikzstyle{bluenode}=[draw,circle,blue,minimum size=0pt,inner sep=2pt,font=\scriptsize]
    \tikzstyle{rednode}=[draw,circle,BrickRed,minimum size=0pt,inner sep=2pt,font=\scriptsize]

\draw (0,0) node[whitenode] (a0)   [] {1};
\draw (1,0) node[whitenode] (a1)   [] {3};
\draw (2,0) node[rednode] (a2)   [] {1};
\draw (3,0) node[whitenode] (a3)   [] {2};
\draw (4,0) node[rednode] (a4)   [] {1};
\draw (5,0) node[rednode] (a5)   [] {2};

\draw (1,1) node[whitenode] (b1)   [] {3};
\draw (2,1) node[whitenode] (b2)   [] {4};
\draw (3,1) node[bluenode] (b3)   [] {5};
\draw (4,1) node[whitenode] (b4)   [] {3};
\draw (5,1) node[rednode] (b5)   [] {4};

\draw (2,2) node[whitenode] (c2)   [] {4};
\draw (3,2) node[bluenode] (c3)   [] {5};
\draw (4,2) node[rednode] (c4)   [] {1};
\draw (5,2) node[rednode] (c5)   [] {4};

\draw (3,3) node[bluenode] (d3)   [] {5};
\draw (4,3) node[rednode] (d4)   [] {1};
\draw (5,3) node[bluenode] (d5)   [] {2};

\draw (4,4) node[bluenode] (e4)   [] {3};
\draw (5,4) node[bluenode] (e5)   [] {2};

\draw (5,5) node[bluenode] (f5)   [] {2};

\draw (a0)--(a1)--(a2)--(a3)--(a4)--(a5);
\draw (b1)--(b2)--(b3)--(b4)--(b5);
\draw (c2)--(c3)--(c4)--(c5);
\draw (d3)--(d4)--(d5);
\draw (e4)--(e5);

\draw (a5)--(b5)--(c5)--(d5)--(e5)--(f5);
\draw (a4)--(b4)--(c4)--(d4)--(e4);
\draw (a3)--(b3)--(c3)--(d3);
\draw (a2)--(b2)--(c2);
\draw (a1)--(b1);

\draw (a0)--(b1)--(c2)--(d3)--(e4)--(f5);
\draw (a1)--(b2)--(c3)--(d4)--(e5);
\draw (a2)--(b3)--(c4)--(d5);
\draw (a3)--(b4)--(c5);
\draw (a4)--(b5);

\draw[thick] (b3)--(b4)--(c4)--(b3);

\end{tikzpicture}
}
\caption{Illustration of the lower bound proof for $k$-set agreement in the $n$-node graph~$\vec{C}_n$ for $k=2$ and $n=5$.  }
\label{fig:triangulation}
\end{figure}

The assignment of a process $\node(\mathbf{x})$ to each vertex $\mathbf{x}$ of the triangulation has also an important property. Let us assume the hypothetical existence of a 1-round algorithm ALG solving 2-set agreement in~$\vec{C}_5$. We color  every vertex $\mathbf{x}$ of the triangulation by the output value of ALG at process $\node(\mathbf{x})$. Such a coloring is displayed on Fig.~\ref{fig:triangulation}(b). The key fact is that, by the validity condition, this coloring must be a Sperner coloring of the triangulation. In the 2-dimensional case displayed on Fig.~\ref{fig:triangulation}, Sperner coloring simply means that 
(1)~the extremities of the triangulation, i.e., the three vertices $a\dots a$ for $a\in\{0,1,2\}$ are colored $0,1$ and $2$, respectively, 
(2)~the sides of the triangulation, namely each of the three paths connecting configuration $a\dots a$ and $b\dots b$ for $a\neq b$ in $\{0,1,2\}$, has each of its vertices colored $a$ or $b$, and 
(3)~the internal vertices can take any color in  $\{0,1,2\}$. 

\subparagraph{Application of Sperner's lemma.}

Sperner's lemma specifically state that any Sperner coloring contains a simplex of dimension~$k$ whose $k+1$ vertices all have different colors in $\{0,\dots,k\}$. On Fig.~\ref{fig:triangulation}, this is for instance the case of the triangle marked in bold. We conclude by showing that the existence of this $k$-dimensional simplex implies the existence of $k+1$ nodes outputting $k+1$ different values for some input configuration, contradicting the specification of agreement. For instance, on Fig.~\ref{fig:triangulation}(b), let us focus on the bold triangle, in which the three colors $0,1$ and $2$ appears (displayed as red, green, and blue) outputted by nodes~$1,3$ and $5$, respectively. This triangle corresponds to the triangle of input configurations $(22110, 21100, 21110)$ on Fig.~\ref{fig:triangulation}(a). Node~1 does not distinguish the three input configurations  $22110, 21100$ and $21110$. Node~3 does not distinguish the two input configurations  $21100$ and $21110$. Therefore, in configuration $21110$, the three nodes $1,3$ and $5$ output $0,1$ and $2$, respectively, contradicting the agreement condition, and thus the possible existence of ALG.

\bigbreak 

The difficulty of the proof is related to the generalization of these arguments to higher dimension, i.e., for arbitrary~$k\geq 2$, and to an arbitrary $n$-node dynamic graph $\mathcal{S}=(G_t)_{t\geq 1}$. In particular, this concerns the specific assignment of input configurations and processing nodes to the vertices of Kuhn's triangulation so that Sperner's Lemma can be applied.

\section{Proof of Theorem~\ref{theo:main}}
\label{sec:lower-bound-k-set-detailed}

Recall that, given a fixed sequence $\mathcal{S}=(G_t)_{t\geq 1}$ of $n$-node graphs on the same set $[n]$ of vertices, for every $r\geq 1$, $H_r$ denotes the $r$-th transitive closure of~$\mathcal{S}$, and $\gamma(H_r)$ denotes the minimum size of a dominating set in~$H_r$. 
The proof is by contradiction. Let $r$ be the smallest positive integer such that $\gamma(H_r)\leq k$, and let us assume that there exists an algorithm ALG solving $k$-set agreement in $\mathcal{S}$ in less than $r$ rounds. 

Let us denote by $\Delta$ the simplex embedded in the $k$-dimensional Euclidean space $\mathbb{R}^k$ defined as the convex hull of the following $k+1$ points of~$\mathbb{R}^k$:
    \[
    (0,\ldots,0),(n,0,\ldots,0),(n,n,0,\ldots,0),\ldots,(n,\ldots,n).
    \]
For instance, for $k=1$, $\Delta$ is the line segment connecting the two points $(0)$ and $(n)$ in $\mathbb{R}$, and,  for $k=2$, $\Delta$ is the triangle connecting the three points $(0,0)$, $(0,n)$, and $(n,n)$ of $\mathbb{R}^2$. For $k=3$, $\Delta$ is a tetrahedron. Let us now recall the Kuhn's triangulation $\tr$ of $\Delta$ (see \cite{chaudhuri1993tight,todd2013computation} for more details). 
    
\subsection{Kuhn's triangulation}

The vertices of the Kuhn's triangulation $\tr$ of $\Delta$ are all the points of $\mathbb{Z}^k$ contained in~$\Delta$. That is, they are the points of $\mathbb{Z}^k$ with coordinates $(x_1,\ldots,x_k)$, where 
\[
n\geq x_1\geq x_2 \geq \ldots \geq x_k \geq 0.
\]
We denote by $V(\tr)$ the vertex set of $\tr$. 
The so-called \emph{primitive} simplices of~$\tr$ (e.g., the primitive triangles of $\tr$, for $k=2$) are defined as follow. Let $\mathbf{e}_1,\ldots,\mathbf{e}_k$ be the unit vectors of $\mathbb{R}^k$, i.e., $\mathbf{e}_i = (0,\ldots,1,\ldots,0)$ with a single $1$ in the $i$-th coordinate, for $i\in\{1,\dots,k\}$. 
A primitive simplex of $\tr$ is determined by a point $\mathbf{y}_0\in \mathbb{Z}^k$ and a permutation $\pi$ of the unit vectors $\mathbf{e}_1,\ldots,\mathbf{e}_k$. Specifically, let us consider the ordering  $\mathbf{e}_{\pi(1)},\ldots,\mathbf{e}_{\pi(k)}$ of the unit vector resulting from this permutation. The primitive simplex of $\tr$ corresponding to~$\pi$ is the convex hull in $\mathbb{R}^k$ of the $k+1$ vertices $\mathbf{y}_0,\ldots,\mathbf{y}_k$, where, for every $i \in \{1,\ldots,k\}$, 
\[
\mathbf{y}_i = \mathbf{y}_{i-1}+\mathbf{e}_{\pi(i)}.
\]
The Kuhn's triangulation $\tr$ contains all primitive simplices that are ``fully inside'' $\Delta$, i.e., for every primitive simplex $\sigma$ defined as above, $\sigma \in \tr$ if and only if $\sigma \cap \Delta =\sigma$.

For instance, Figure~\ref{fig:triangulation} displays Kuhn's triangulation of $\Delta$ in the case $k=2$ and $n=5$, i.e., when $\Delta$ is the convex hull of the three points $(0,0),(5,0)$, and $(5,5)$. For $k=3$, Figure~\ref{fig:primitive} displays all primitive simplices of $\tr$ that fall inside a unit cube in $\mathbb{R}^3$. E.g., the primitive simplex $\sigma=ABCF$ is defined by vertex $A$ (with coordinate equal to some $\mathbf{y}_0\in\mathbb{Z}^3$), and the permutation $e_1,e_3,e_2$ of the unit vectors $e_1,e_2,e_3$. That is, the simplex $\sigma$ is the convex hull of the vertices $\mathbf{y}_0, \mathbf{y}_1 = \mathbf{y}_0+\mathbf{e}_1, \mathbf{y}_2 = \mathbf{y}_1+\mathbf{e}_3$, and $\mathbf{y}_3 = \mathbf{y}_2+\mathbf{e}_2$.

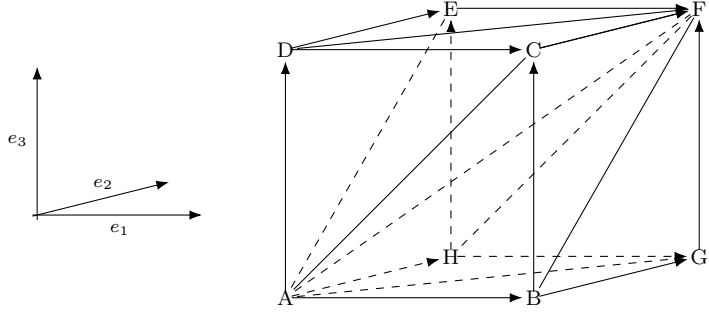
\begin{figure}[tb]
    \centering
    \centering
\begin{tikzpicture}[scale = 1.1]
    \tikzstyle{whitenode}=[circle,fill=white,minimum size=0pt,inner sep=0pt,font=\footnotesize]

    \draw (0,0) node[whitenode] (a)   [] {A};
    \draw (3,0) node[whitenode] (b)   [] {B};
    \draw (3,3) node[whitenode] (c)   [] {C};
    \draw (0,3) node[whitenode] (d)   [] {D};
    \draw (5,0.5) node[whitenode] (g)   [] {G};
    \draw (5,3.5) node[whitenode] (f)   [] {F};
    \draw (2,3.5) node[whitenode] (e)   [] {E};
    \draw (2,0.5) node[whitenode] (h)   [] {H};

    \draw[arr]  (a)--(b);
    \draw[arr] (b)--(c);
    \draw[arr] (c)--(f);
    \draw[arr] (a)--(d);
    \draw[arr] (d)--(c);
    \draw[arr] (c)--(f);
    \draw[arr] (g)--(f);
    \draw[arr] (b)--(g);
    \draw [arr](d)--(e);
    \draw [arr] (e)--(f);
    \draw [dashed,arr] (a)--(h);
    \draw [dashed,arr] (h)--(g);
    \draw [dashed,arr] (h)--(e);

    \draw [dashed] (a)--(f);
    \draw (a)--(c);
    \draw (f)--(b);
    \draw (d)--(f);
    \draw [dashed] (a)--(g);
    \draw [dashed] (f)--(h);
    \draw [dashed] (a)--(e);
    
    \draw (-3,1) node[whitenode] (o) [] {};
    \draw (-1,1) node[whitenode] (x) [] {};
    \draw (-3,2.8) node[whitenode] (z) [] {};
    \draw (-1.4,1.4) node[whitenode] (y) [] {};

    \draw[arr] (o)--(x) node[draw=none,fill=none,font=\scriptsize,midway,below] {$e_1$};
    \draw[arr] (o)--(y) node[draw=none,fill=none,font=\scriptsize,midway,above] {$e_2$};
    \draw[arr] (o)--(z) node[draw=none,fill=none,font=\scriptsize,midway,left] {$e_3$};
\end{tikzpicture}
    \caption{The six primitive simplices that are inside a unit cube in $\mathbb{R}^3$.  
    $ABGF = \{A,\mathbf{e}_1,\mathbf{e}_2,\mathbf{e}_3\}$, 
    $ABCF = \{A,\mathbf{e}_1,\mathbf{e}_3,\mathbf{e}_2\}$, 
    $AHGF = \{A,\mathbf{e}_2,\mathbf{e}_1,\mathbf{e}_3\}$, 
    $AHEF = \{A,\mathbf{e}_2,\mathbf{e}_3,\mathbf{e}_1\}$, 
    $ADCF = \{A,\mathbf{e}_3,\mathbf{e}_1,\mathbf{e}_2\}$, and 
    $ADEF = \{A,\mathbf{e}_3,\mathbf{e}_2,\mathbf{e}_1\}$.}
    \label{fig:primitive}
\end{figure}

\subsection{Input assignment} 

An input configuration of $k$-set agreement is a vector $(v_1,\dots,v_n)$, where $v_i\in\{0,\dots,k\}$ is the input value of node~$i\in[n]$. Unless it may create confusion, we allow ourselves to remove the vector notation for a configuration. For example, for $n=4$ and $k=2$, we merely write $2101$ for the input configuration $(2,1,0,1)$. For each vertex $\mathbf{x}=(x_1,\ldots, x_k)$ of $\tr$, we assign an input configuration denoted by $\inp(\mathbf{x})$ to the $n$ processing nodes, as follow. The $x_k$ nodes $1,\dots,x_k$ have input~$k$, 
the $x_{k-1}-x_k$ nodes $x_k+1,\dots,x_{k-1}$ have input~$k-1$, 
the $x_{k-2}-x_{k-1}$ nodes $x_{k-1}+1,\dots,x_{k-2}$ have input~$k-2$, 
etc., 
the $x_1-x_2$ nodes $x_2+1,\dots,x_1$ have input~$1$, and 
the remaining $n-x_1$ nodes $x_1+1,\dots,n$ have input~$0$. 
For example, for $n=5$ and $k=3$, if $\mathbf{x}=(3,3,1)$, then $\inp(\mathbf{x}) = 32200$. For $n=5, k=2$, Figure~\ref{fig:triangulation}(a) presents the input configurations assigned to the vertices of the triangulation $\tr$ of the simplex $\Delta$ defined as the convex hull of three points $(0,0),(0,5),(5,5)$ in $\mathbb{Z}^2$.

By construction of the input assignments to the vertices of $\tr$, we get the following. 

\begin{lemma*}\label{lem:inpAssignment}
    Let $\mathbf{y}_0,\ldots,\mathbf{y}_k$ be a primitive simplex  of~$\tr$. The set  $\{\inp(\mathbf{y}_i) \mid i\in\{0,\ldots,k\}\}$ of input configurations satisfies the following:
    \begin{itemize}
        \item For every  $i\in\{0,\ldots,k-1\}$, there is exactly one node with different input values in the two configurations $\inp(\mathbf{y}_i)$ and $\inp(\mathbf{y}_{i+1})$.
        \item For every  $i\in\{1,\ldots,k\}$, there are exactly $i$ nodes with different input values in the two configurations $\inp(\mathbf{y}_i)$ and $\inp(\mathbf{y}_0)$. Moreover, if $V_i$ denotes the set of nodes with different input values in $\inp(\mathbf{y}_i)$ and $\inp(\mathbf{y}_0)$, then $V_i \subset V_{i+1}$, and $ |V_{i+1} \smallsetminus V_i| =1$.
    \end{itemize}
\end{lemma*}

\begin{proof}
    For each $i\in\{1,\dots,k\}$, there is an unit vector $\mathbf{e}_j \in \mathbb{R}^k$ such that $\mathbf{y}_{i+1} = \mathbf{y}_i + \mathbf{e}_j$. Let $i\in\{1,\dots,k\}$, and let $\mathbf{y}_i=(x_1,\ldots,x_n)$. We get $\mathbf{y}_{i+1}=(x_1,\ldots,x_j-1,x_j+1,x_{j+1},\ldots,x_n)$.
    By the definition of the input assignment to vertices $\mathbf{y}_i$ and $\mathbf{y}_{i+1}$, the unique node with different inputs in $\inp(\mathbf{y}_i)$ and $\inp(\mathbf{y}_{i+1})$ thus node $x_j+1$. 
    
    Regarding the second property in the statement of the lemma, let $X \subseteq [k]$ be the index set such that 
    $
    \mathbf{y}_i = \mathbf{y}_0 + \sum_{i \in X} \mathbf{e}_i.
    $
    There exists $j \in [k] \smallsetminus X$ such that 
    \[
    \mathbf{y}_{i+1} = \mathbf{y}_i+\mathbf{e}_j = \mathbf{y}_0 + \sum_{i \in X \cup \{j\}} \mathbf{e}_i.
    \]
    The set of processes with different inputs in $\inp(\mathbf{y}_i)$ and $\inp(\mathbf{y}_0)$ is $V_i = \{x_i \mid i \in X\}$.
    The set of nodes with different inputs in $\inp(\mathbf{y}_{i+1})$ and $\inp(\mathbf{y}_0)$ is $V_{i+1} = \{x_i \mid i \in X\cup\{j\} \}$. Therefore, $V_i \subseteq V_{i+1}$, and $ |V_{i+1} \smallsetminus V_i| \leq 1$, as claimed. \qed
\end{proof}

\begin{corollary*} \label{corollary:inpAssignment}
    For every vertex $\mathbf{y}=(x_1,\ldots,x_k)$ of~$\tr$, for every subset $X \subseteq [k]$ such that $\mathbf{z} = \mathbf{y}-\sum_{i\in X} \mathbf{e_i}$ belongs to $V(\tr)$, the set of nodes with different input values in $inp(\mathbf{y})$ and $inp(\mathbf{z})$ is $\{x_i \;\mid\; i \in X\}$.
\end{corollary*} 

\subsection{Process assignment}

We assign to each vertex $\mathbf{x} \in V(\tr)$ a processing node in~$[n]$, denoted by $\node(\mathbf{x})$, as follows. Let $\mathbf{x}=(x_1,\ldots,x_k)$. We consider the processing nodes $\{x_1,\ldots,x_k\} \smallsetminus\{0\}$. Note that it may be the case that $x_i = x_{i+1}$ for some~$i$, and thus $|\{x_1,\dots,x_k\}|$ might be smaller than~$k$.  Since $\gamma(H_r) >k$, there exists at least one node that does not receive information from any process $\{x_1,\dots,x_k\} \smallsetminus \{0\}$ after $r$~rounds. We let $\node(\mathbf{x})$ be any such node. For example, in case $n=5, k=2$, we refer to Figure~\ref{fig:triangulation}(b) as an illustration of the process assignment to the vertices of the triangulation $\tr$ of the simplex $\Delta$ defined as the convex hull of points $(0,0),(5,0),(5,5)$ in~$\mathbb{Z}^2$. 

Let us sort the vertices of $\tr$ in alphabetical order, i.e.,  
\[
\mathbf{x} = (x_1,\ldots,x_k) \prec \mathbf{y} = (y_1,\ldots,y_k)
\]
if there is $i\in\{1,\dots,k\}$ such that $x_j = y_j$ for all $j \in \{1,\ldots,i\}$, and $x_{i+1}<y_{i+1}$. Note that, by definition of~$\tr$, $\prec$ is a total order on the vertex set of~$\tr$. 
Recall that a unit hypercube $Q$ in $\mathbb{R}^k$ is the Cartesian product of $k$ unit intervals, i.e., 
\[
Q=[\mathbf{y},\mathbf{y}+\mathbf{e}_1]\times \ldots \times [\mathbf{y},\mathbf{y}+\mathbf{e}_k]
\]
for a point $\mathbf{y} \in \mathbb{R}^k$.
Observe that, for every vertex $\mathbf{x} \in V(\tr)$, there is one hypercube $Q$ such that $\mathbf{x}$ is the maximum vertex in that hypercube with respect to~$\prec$. Given this hypercube $Q$, let us denote by $\mathbf{y}$ the minimum vertex in~$Q$. We have $\mathbf{x} = \mathbf{y}+ \sum_{i=1}^k\mathbf{e}_i$. By Corollary~\ref{corollary:inpAssignment}, there are at most $k$ nodes with different inputs in $\inp(\mathbf{x})$ and $\inp(\mathbf{y})$. The process $\node(\mathbf{x})\in [n]$ cannot perceive the difference between $\inp(\mathbf{x})$ and $\inp(\mathbf{y})$ in less than $r$ rounds of communication. Indeed, by corollary~\ref{corollary:inpAssignment}, for every set $X \subseteq [k]$, if $\mathbf{z} = \mathbf{x}-\sum_{i\in X}\mathbf{e}_i$ in $V(\tr)$, then $\node(\mathbf{x})$ cannot see the difference between $\inp(\mathbf{x})$ and $\inp(\mathbf{z})$ in less than $r$ rounds of communication. Therefore, for every vertex $\mathbf{z}$ of the hypercube~$Q$, $\node(\mathbf{x})$ cannot perceive the difference between $\inp(\mathbf{x})$ and $\inp(\mathbf{z})$  in less than $r$ rounds of communication. 

\begin{lemma*}\label{lem:view}
    Let $\mathbf{y}_0,\ldots,\mathbf{y}_k$ be a primitive simplex of $\tr$. After $(r-1)$ rounds of communication, for every $i\in\{1,\dots,k\}$, $\node(\mathbf{y}_i)$ has the same view in $\inp(\mathbf{y}_i)$ and $\inp(\mathbf{y}_0)$.
\end{lemma*}

\begin{proof}
    Let $i\in\{1,\dots,k\}$, and let $\mathbf{y}_i=(x_1,\ldots,x_n)$. Let $X \subseteq [k]$ be an index set such that $\mathbf{y}_i = \mathbf{y}_0 + \sum_{i\in X} \mathbf{e}_i$. By corollary~\ref{corollary:inpAssignment}, the set of processes with 
    different inputs in $\inp(\mathbf{y}_i)$ and $\inp(\mathbf{y}_0)$ is $V_i = \{x_i \mid i \in X \}$. Note that for every $i \in X, x_i>0$. Since $\node(\mathbf{y}_i)$ does not receive information from any process in $\{x_1,\ldots, x_k\} \smallsetminus \{0\}$ in less than $r$ rounds, it follows that $\node(\mathbf{y}_i)$ has the same view in $\inp(\mathbf{y}_i)$ and $\inp(\mathbf{y}_0)$. \qed
\end{proof}

\noindent
We now have all ingredients to complete the proof. Let us define a coloring 
\[
c: V(\tr) \rightarrow \{0,\ldots,k\}
\]
as follow. For every vertex $\mathbf{x}$ of~$\tr$, we set 
\begin{equation}\label{eq:our-sperner-coloring}
c(\mathbf{x})= \mbox{ALG}(\node(\mathbf{x}),\inp(\mathbf{x}),r),
\end{equation}
that is, $\mathbf{x}$ is colored by the output of ALG at $\node(\mathbf{x})$ in the input configuration $\inp(\mathbf{x})$ after $(r-1)$ rounds. This coloring is defined so that Sperner's Lemma can be applied to complete the proof. 

\subsection{Sperner Coloring}

 Recall that a \emph{$d$-simplex} $\Delta^d$ is the convex hull of $d+1$ points $v_1,\ldots,v_{d+1} \in \mathbb{R}^d$ that are affinely independent. The vertices of $\Delta^d$ is 
 \[
 V(\Delta^d) = \{v_1,\ldots,v_{d+1} \}.
 \]
 A \emph{simplicial subdivision} (a.k.a.~triangulation) of a $d$-simplex $\Delta^d$ is a finite set of $d$-simplicies $\mathcal{S} = \{\sigma_1,\ldots,\sigma_m\}$ such that $\Delta^d= \cup_{1}^m \sigma_i$, and, for every $i \neq  j$, $\sigma_i\cap \sigma_j$ is either empty, or is a  face common to both $S_i$ and $S_j$. The vertices of a subdivision are the vertices of all its $d$-simplicies, that is, $V(\mathbf{S}) = \cup_{i=1}^mV(S_i)$. Also recall that, given a simplicial subdivision $\mathcal{S}$ of a $d$-simplex $\Delta^d$, where $V(\Delta^d) = \{v_1,\ldots,v_{d+1}\}$, a Sperner coloring is a map 
 \[
 f: V(\mathcal{S}) \rightarrow \{1,\ldots,d+1\},
 \]
 such that, for every $I \subseteq [d+1]$, whenever a vertex $u \in V(\mathcal{S})$ lies on the face $F$ of $\Delta^d$ induced by~$I$, i.e., on the convex hull of $\{v_i \mid i\in I\}$, we have $f(u) \in I$.

\begin{lemma*}[Sperner's lemma] \label{lem:Sperner} 
Every Sperner coloring of a simplicial subdivision of a $d$-simplex contains a simplex whose $d+1$ vertices are colored with different colors.  
\end{lemma*}

To apply Sperner's lemma on  $\tr$, we show that $c$ is indeed a Sperner coloring of $\tr$. 

\begin{lemma*}\label{lem:SpernerCol}
    The coloring $c:V(\tr) \rightarrow \{0,\ldots,k\}$ defined in Eq.~\eqref{eq:our-sperner-coloring} is a Sperner coloring of~$\tr$. 
\end{lemma*}

\begin{proof}
    For every $i\in \{0,\ldots,k\}$, let $\mathbf{v_i} \in \mathbb{Z}^k$ be a vertex of the simplex $\Delta$, i.e., $\mathbf{v_i}$ has coordinate $(n,\ldots,n,0,\ldots,0)$ in which the first $i$ coordinates are equal to $n$, and the remaining coordinates are equalt to~$0$. For every $i \in \{0,\ldots,k\}$, all nodes have input $i$ in $\inp(\mathbf{v_i})$, so $c(\mathbf{v_i})=i$ because of the validity condition.
    
    Let $\mathbf{v} = (x_1,\ldots,x_k)$ be a vertex of $V(\tr)$ that lies on a face $F$ of $\Delta$. We prove that $c(\mathbf{v})$ is one of the colours of the vertices of~$F$.
    Vertex $\mathbf{v}$ can be represented as a convex combination of $\mathbf{v_i}, i\in \{0,\ldots,k\}$, that is:
    \[
    \mathbf{v} = \sum_{i=0}^k \lambda_j \mathbf{v_j},
    \]
    in which $\lambda_i \geq 0$ for every $i\in\{0,\dots,k\}$, and $ \sum_{i=0}^k \lambda_i =1$. Since $\mathbf{v_j}=n\sum_{i=0}^j \mathbf{e_i}$, we have 
    \[
    \mathbf{v} =  n\sum_{j=0}^k \left(\lambda_j\sum_{i=0}^j \mathbf{e_i} \right)
    = n\sum_{i=0}^k \left(\sum_{j=i}^n \lambda_j\right) \mathbf{e_i}.
    \]
    Therefore,  for every $i \in \{0,\ldots,k\}$, $x_i = n \sum_{j=i}^n \lambda_j$, and $x_{i}-x_{i+1} = n\lambda_i$ whenever $i<k$. Since $F$ is a face  of $\Delta$, there exists an index set $I \subseteq \{0,\ldots,k\}$ such that $F$ is the convex hull of $\{\mathbf{v_i} \mid i \in I\}$. Since $\mathbf{v}$ lies in face $F$, we get thatn for every $j \notin I$, $\lambda_j =0$, and thus $x_j-x_{j+1} = 0$. Therefore, for every $j \notin I$, no node has input $j$ in $\inp(v)$. By validity condition of $k$-set agreement task, $c(\mathbf{v})\in I$, so is indeed one of the colors of the vertex $\mathbf{v_i}$ of face~$F$. \qed
\end{proof}

\subsection{Putting Things Together}

To complete the proof, it is sufficient to show that there is an input configuration of $k$-set agreement in which the nodes decide $k+1$ different output value when performing ALG, which is a contradiction with the hypothesis that ALG solves $k$-set agreement (the agreement condition is violated). 

For every primitive simplex of $\tr$, there is a minimum vertex $\mathbf{y}_0$, and a permutation $\mathbf{e}_{\pi(1)},\ldots,\mathbf{e}_{\pi(k)}$ of the $k$ unit vectors that are defining that simplex. By lemma \ref{lem:SpernerCol}, the coloring $c$ is a Sperner coloring of $\tr$, and thus, thanks to Sperner's lemma (cf. Lemma~\ref{lem:Sperner}), there exists a primitive simplex $\mathbf{y}_0,\ldots,\mathbf{y}_k$ of $\tr$ such that 
\[
\{c(\mathbf{y}_0),\ldots,c(\mathbf{y}_k)\}=\{0,\dots,k\},
\]
i.e., for every $i$, the vertex $\mathbf{y}_i$ has a color different from the colors of all the other vertices in the simplex $\mathbf{y}_0,\ldots,\mathbf{y}_k$. By Lemma~\ref{lem:view}, this implies that, for every $i\in \{0,\dots,k\}$, $\node(\mathbf{y}_i)$ has the same view in $\inp(\mathbf{y}_i)$ and $\inp(\mathbf{y}_0)$. Therefore, for every $i\in\{1,\ldots,k\}$, 
\[
\mbox{ALG}(\node(\mathbf{y}_i),\inp(\mathbf{y}_i),r)=\mbox{ALG}(\node(\mathbf{y}_i),\inp(\mathbf{y}_0),r).
\]
As a consequence, $\node(\mathbf{y}_0),\ldots,\node(\mathbf{y}_k)$ are deciding $k+1$ different values in the input configuration $\inp(\mathbf{y}_0)$. This completes the proof of Theorem~\ref{theo:main}. \qed

\section{Conclusion}

We have re-proved the lower bound for $k$-set agreement in the $\KA$ model originally established in~\cite{castaneda2021topological}, in a way that we hope is accessible to a wider audience. 
Our new proof may be of independent interest, and,  for instance, the case of 2-set agreement for 5~processes may even serve as an elegant, self-contained example of the application of topological methods in distributed computing, which can be readily taught in class. 

One question remains: would it be possible to avoid using any tools from algebraic topology for establishing non-trivial lower bounds on the complexity of $k$-set agreement? 
In this paper, we limit the usage of algebraic tools to Sperner's Lemma, yet this lemma is quite a strong tool (cf. its connection to Brouwer’s fixed point theorem). 
The question remains whether more ``basic'' arguments could be used, or if $k$-set agreement is somehow intrinsically related to topology in a way that still requires to be formalized. We refer to \cite{AlistarhAEGZ23,AttiyaCR23,AttiyaFPR23} for recent advances on this question.

\section*{Acknowledgement} The authors thank Stephan Felber, Mikaël Rabie, Hugo Rincon Galeana and Ulrich Schmid for fruitful discussions on this paper.

\bibliography{ref}

\end{document}